\documentclass[a4paper,11 pt,twoside]{amsart}

\usepackage[english]{babel}
\usepackage[utf8]{inputenc}

\usepackage{times}
\usepackage{amsmath}
\usepackage{amsfonts}
\usepackage{amssymb}
\usepackage{amsmath}
\usepackage{amsthm}
\usepackage{graphicx}
\usepackage{array}
\usepackage{color}
\usepackage{mathrsfs}
\usepackage{hyperref}
\usepackage{eucal}
\usepackage{esint} %%% for \fint
\usepackage{tikz}
\usepackage{upgreek}
\usepackage{enumitem}

\allowdisplaybreaks

\theoremstyle{plain}
\newtheorem{theorem}{Theorem}[section]
\theoremstyle{plain}
\newtheorem{corollary}{Corollary}[section]
\theoremstyle{plain}
\newtheorem{proposition}{Proposition}[section]
\theoremstyle{plain}
\newtheorem{lemma}{Lemma}[section]
\theoremstyle{definition}
\newtheorem{definition}{Definition}[section]
\theoremstyle{definition}

\theoremstyle{definition}

\theoremstyle{definition}

\theoremstyle{definition}

\theoremstyle{definition}
\newtheorem{remark}{Remark}[section]

\numberwithin{equation}{section}
\numberwithin{figure}{section}
\numberwithin{table}{section}

\newcommand{\R}{\mathbb{R}}

\newcommand{\C}{\mathbb{C}}

\newcommand{\s}[1]{\CMcal{#1}}
                  
\newcommand{\bb}[1]{\mathscr{#1}}

\newcommand{\expo}[1]{\,\mathrm{e}^{#1}\,}                 

\newcommand{\dd}{\,\mathrm{d}}
\newcommand{ \ii}{\,\mathrm{i}\,}

\newcommand{\virg}[1]{\lq\lq#1\rq\rq}                \newcommand{\ie}{\textsl{i.\,e.\,}}

\newcommand{\cf}{\textsl{cf}.\,}

\newcommand{\lp}{\left(}
\newcommand{\rp}{\right)}
\newcommand{\lv}{\left\Vert}
\newcommand{\rv}{\right\Vert}

\providecommand{\abs}[1]{\left\vert#1\right\vert}
\providecommand{\norm}[1]{\left\Vert#1\right\Vert}

\hypersetup{
citecolor=blue,
colorlinks=true, %colorise les liens
breaklinks=true, %permet le retour a la ligne dans les liens trop longs
urlcolor= blue, %couleur des hyperliens
linkcolor= black, %couleur des liens internes
bookmarksopen=true, %si les signets Acrobat sont crees,
pdftitle={Anderson Hamiltonian}, %informations apparaissant dans
}

\begin{document}

\title[Perturbation Theory for the Thermal Hamiltonian: 1D Case]{Perturbation Theory for the Thermal Hamiltonian: 1D Case}

\author[G. De~Nittis]{Giuseppe De Nittis}

\address[G. De~Nittis]{Facultad de Matem\'aticas \& Instituto de F\'{\i}sica,
  Pontificia Universidad Cat\'olica de Chile,
  Santiago, Chile.}
\email{gidenittis@mat.uc.cl}

\author[V. Lenz]{Vicente Lenz}

\address[V. Lenz]{Departamento de Matem\'aticas, Facultad de Ciencias,  Universidad de Chile, Santiago, Chile}
\email{vicente.lenz@ug.uchile.cl}

\vspace{2mm}

\date{\today}

\maketitle

\begin{abstract}
This work   continues the study of the \emph{thermal Hamiltonian}, initially proposed by J. M. Luttinger in 1964 as a model for the conduction of thermal currents in solids. The previous work \cite{denittis-lenz-20} contains  a complete study of the \virg{free} model in one spatial dimension  along with a preliminary scattering result for convolution-type perturbations. This work complements the results obtained in \cite{denittis-lenz-20} by providing a detailed  analysis  of the perturbation theory for the one-dimensional thermal Hamiltonian. In more detail the following result are established: the regularity and decay properties for elements in the domain of the unperturbed thermal Hamiltonian;
the determination of a class of self-adjoint and relatively  compact perturbations of the thermal Hamiltonian; the proof of the existence and completeness of wave operators for a subclass of such potentials.

\medskip

\end{abstract}

\tableofcontents

%--------------------%
\section{Introduction}\label{sect:intro}
In order to study the  \emph{thermal transport} in the matter, J. M. Luttinger  proposed in 1964 a model which allows  a \virg{mechanical} derivation of the thermal coefficients \cite{luttinger-64}.
Such a model has been eventually studied and generalized successfully in various later works such as   \cite{smrcka-streda-77,vafek-melikyan-tesanovic-01}. The essential insight of the Luttinger's model is to describe the effect of the thermal gradient in the matter by a \emph{fictitious} gravitational field which affects the dynamics of a charged particle moving in a \emph{background} material. 

\medskip

In 
 absence of thermal fields, and ignoring all physical constants,  the dynamics of a one-dimensional quantum particle is described by the Hamiltonian
\begin{equation}\label{eq:intro_01}
h_V\;:=\;p^2\;+\;V
\end{equation}
where $p:=-\ii\frac{\dd}{\dd x}$ is the \emph{momentum operator} and $V$
is the \emph{background} (or
electrostatic) \emph{potential} which takes care of the interaction of the particle
with the atomic structure of the matter. 
In the absence of interaction with matter ($V=0$) the dynamics is simply described by $h_0=p^2$. The effect of the thermal field is introduced in the model by a \emph{thermal potential} which is proportional to the local content of energy. Since the latter is given by the Hamiltonian \eqref{eq:intro_01} itself, one ends with the following
(effective) model
\begin{equation}\label{eq:intro_02}
H_{T,V}\;:=\;h_V\;+\;\frac{\lambda}{2}\{h_V,x\}
\end{equation}
known  as \emph{thermal Hamiltonian} or  \emph{Luttinger's Hamiltonian}. We will refer to \cite[Section 1.1]{denittis-lenz-20}, and references therein,
for more details on the physical justification of \eqref{eq:intro_02}. Here,  it is worth to point out that mathematically the {thermal potential}
is introduced by the {anti-commutator} $\{\;,\;\}$ between $h_V$ and the \emph{position operator} $x$, and that the
 parameter $\lambda>0$  describes the strength of the thermal field.
 
 \medskip

The Hamiltonian \eqref{eq:intro_02} can be rearranged in the form
\begin{equation}\label{eq:intro_03}
H_{T,V}\;:=\;H_T\;+\;W_V
\end{equation}
where
\begin{equation}\label{eq:intro_04}
H_{T}\;:=\;h_0\;+\;\frac{\lambda}{2}\{h_0,x\}
\end{equation}
is the \emph{thermal Hamiltonian} in absence of a background  potential
and 
\begin{equation}\label{eq:intro_05}
W_V(x)\;:=\;(1+\lambda x)V(x)\;,\qquad x\in\R
\end{equation}
is the resulting potential that combines the effects of the thermal field and the electrostatic interaction with the matter.
The study of the spectral properties of the Hamiltonian $H_T$ has been the central argument of \cite{denittis-lenz-20}. The main aim of this work is to provide  a satisfactory description of the spectral theory  for the perturbed Hamiltonian $H_{T,V}$ and to derive the scattering theory \cite{reed-simon-III,kato-95,yafaev-92} for the pair $(H_{T},H_{T,V})$
for a sufficiently general class of background potentials $V$.

\medskip

Before describing the new results, let us recall some essential facts about the \virg{unperturbed} operator $H_T$. On sufficiently regular functions $\psi:\R\to\C$ the operator $H_T$ acts as follows
\begin{equation}\label{eq:intro_06}
(H_T\psi)(x)\;=\;-(1+\lambda x)\psi''(x)\;-\;\lambda\psi'(x)
\end{equation}
where $\psi'$ and $\psi''$ are the first and the second derivatives of $\psi$, respectively. However, when restricted to the Schwartz space $\s{S}(\R)$, $H_T$ turns out to be symmetric but not essentially self-adjoint. In fact, the operator initially defined on $\s{S}(\R)$ by \eqref{eq:intro_06} admits a one-parameter family of  self-adjoint extensions. However, it turns out that all these self-adjoint extensions are unitarily equivalent and, without loss of generality, one can focus on a specific \virg{canonical} realization 
  \cite[Theorem 1.1]{denittis-lenz-20}. Such a realization is obtained by considering the dense domain
\begin{equation}\label{eq:intro_07}
\s{D}_{T,0}\;:=\;\s{S}(\R)\;+\;\C[\kappa_0]
\end{equation}
 and the prescription 
\begin{equation}\label{eq:intro_08}
\big(H_T(\psi+c\kappa_0)\big)(x)\;=\;(H_T\psi)(x)\;+\;c\lambda\kappa_1(x)\;,\qquad\psi\in \s{S}(\R)
\end{equation}
where the term $H_T\psi$ is given by \eqref{eq:intro_06} and 
\begin{equation}\label{eq:act_HT_02}
\begin{aligned}
\kappa_0(x)\;:&=\;-\sqrt{\frac{8}{\pi}}\;{\rm sgn}\left(x+\frac{1}{\lambda}\right)\;{\rm kei}\left(2\sqrt{\left|x+\frac{1}{\lambda}\right|} \right)\\
\kappa_1(x)\;:&=\;\sqrt{\frac{8}{\pi}}\;{\rm ker}\left(2\sqrt{\left|x+\frac{1}{\lambda}\right|} \right)\;.
\end{aligned}
\end{equation}
In \eqref{eq:act_HT_02}  ${\rm kei}$ and ${\rm ker}$  denote the \emph{irregular Kelvin functions} of $0$-th order (\cf \cite[Chap. 55]{oldham-myland-spanier-09} or \cite[Sect. 10.61]{nist}) while the \emph{sign function} is defined by ${\rm sgn}(x):=x/|x|$ if $x\neq0$ and  ${\rm sgn}(0):=0$. 
It turns out that the operator defined by \eqref{eq:intro_08} is essentially self-adjoint on the domain \eqref{eq:intro_07},
and this fact provides a rigorous definition for the thermal Hamiltonian.
  \begin{definition}[1-D {unperturbed} thermal Hamiltonian]\label{def:unp-TH}
  The \emph{{unperturbed} thermal Hamiltonian}, still denoted with 
  $H_T$, is the self-adjoint operator on $L^2(\R)$ defined by \eqref{eq:intro_07} and \eqref{eq:intro_08} on the domain
  $$
  \s{D}_{T}\;:=\;\overline{\s{D}_{T,0}}^{\;\|\;\|_{H_T}}
  $$
obtained by the closure of $\s{D}_{T,0}$ with respect to the graph norm $\|\psi\|_{H_T}:=\|\psi\|_{L^2}+\|H_T\psi\|_{L^2}$.
  
   \end{definition}

\medskip

In view of \cite[Theorem 1.1]{denittis-lenz-20} we  know that
$H_T$ has a purely absolutely continuous spectrum given by 
\begin{equation}\label{eq:intro_010a}
\sigma\big(H_{T}\big)\;=\;\sigma_{\rm a.c.}\big(H_{T}\big)\;=\;\R\;
\end{equation}
independently of $\lambda>0$.

\medskip

We are now in position to present the main results of this work. For that, let us recall the definition of the \emph{critical point} 
\begin{equation}\label{eq:intro_010}
x_c\;\equiv\;x_c(\lambda)\;:=\;-\lambda^{-1}\;
\end{equation}
which plays an important role in the singular behavior of the dynamics generated by the 
{unperturbed} operator $H_{T}$  \cite{denittis-lenz-20}.
The first result concerns the determination of a class of self-adjoint perturbation of $H_T$.
\begin{theorem}[Self-adjoint perturbations]\label{theo:01}
Let $H_T$ be the {unperturbed} thermal Hamiltonian described in Definition \ref{def:unp-TH}. Let $V:\R\to\R$ be a background potential such that
$$
V(x)\;=\;\;\frac{V_1(x)}{|x-x_c|^{\frac{3}{4}}}\;+\;\frac{V_2(x)}{|x-x_c|}
$$ 
with $V_1\in L^2(\R)$ and $V_2\in L^\infty(\R)$. Then, the \emph{perturbed} thermal Hamiltonian $H_{T,V}$ given by \eqref{eq:intro_03}, with potential 
$W_V$ given by \eqref{eq:intro_05}, is self-adjoint on the domain $ \s{D}_{T}$.
\end{theorem}

\begin{remark}
Although Theorem \ref{theo:01} stipulates that $H_{T,V}$ is self-adjoint for a large class of background potentials $V$, from a physical point of view this result is not yet totally satisfactory. In fact, the standard model for the dynamics of a charged particle in a (semi-)metal is $h_{V_{\rm per}}:=p^2+V_{\rm per}$ with $V_{\rm per}$ a \emph{periodic} background potential. However, every
$V_{\rm per}\neq 0$  does not meet the conditions of Theorem \ref{theo:01}, and as a consequence the question of the self-adjointness of $H_{T,V_{\rm per}}$  remains open. This is not an irrelevant fact  since  $H_{T,V_{\rm per}}$ is the relevant model (tacitally)
considered in \cite{smrcka-streda-77,vafek-melikyan-tesanovic-01} for the derivation of thermal conductivity in condensed matter systems. It is also worth noting that  Theorem \ref{theo:01} allows background potentials which are singular around the {critical point} $x_c$.
   \hfill $\blacktriangleleft$
\end{remark}

\medskip

The second main result describes a class of relatively compact perturbations. For that we will need to introduce the  family of resolvents 
$$
R_z(H_T)\;:=\;(H_T- z{\bf 1})^{-1}\;,\qquad z\in\C\setminus \R\;.
$$
\begin{theorem}[Relatively compact perturbations]\label{theo:02}
Let $H_T$ be the {unperturbed} thermal Hamiltonian described in Definition \ref{def:unp-TH}. Let $V:\R\to\R$ be a background potential such that
$$
V(x)\;=\;\;\frac{V_1(x)}{|x-x_c|^{\frac{3}{4}}}
$$ 
with $V_1\in L^2(\R)$.
Then, $W_V R_z(H_T)$ is Hilbert-Schmidt (hence compact) for every $z\in\C\setminus \R$, where the potential $W_V$ is given by \eqref{eq:intro_05}.
\end{theorem}

\medskip

By combining Theorem \ref{theo:02} with the Weyl Theorem about the stability of the essential spectrum \cite[Theorem XIII.14]{reed-simon-IV} one obtains the following result.
\begin{corollary}[Essential spectrum]\label{cor:intro1}
Let $V$ be a background potential as in Theorem \ref{theo:02}. Then, the essential spectrum of the {perturbed} thermal Hamiltonian $H_{T,V}$ is $
\sigma_{\rm ess}(H_{T,V})=\sigma_{\rm ess}(H_{T})=\R
$.
\end{corollary}

\medskip

The question of the existence of \emph{embedded eigenvalues} is not answered by  Corollary \ref{cor:intro1} and is left open for future investigations.

\medskip

The final result concerns the scattering theory for the pair $(H_{T},H_{T,V})$. Let us recall the (formal) definition of the \emph{wave operators} \cite[Section XI.3]{reed-simon-III}
\begin{equation}\label{eq:intro_011}
			\Omega_\pm(V)\;:=\;{\rm s}-\lim_{t\rightarrow \mp \infty}\expo{\ii H_{T,V}t}\expo{-\ii H_{T}t}\;,
		\end{equation}
where the limits are meant in the strong operator topology. 
It is worth noting that in the definition \eqref{eq:intro_011} we have tacitly used the fact that the spectral projection on the absolutely continuous part of the spectrum of $H_{T}$ coincides with the identity in view of \eqref{eq:intro_010a}.
The \emph{scattering matrix} is defined by
\begin{equation}
			\bb{S}(V)\;:=\;\Omega_-(V)^*\Omega_+(V)\;.
		\end{equation}
\begin{theorem}[Scattering theory]\label{theo:03}
 Let $V:\R\to\R$ be a background potential such that bot $V$ and $|V|^{\frac{1}{2}}$ satisfy the conditions of Theorem \ref{theo:02}. Then, the wave operators $\Omega_\pm(V)$ exist and the scattering matrix $\bb{S}(V)$ is unitary.
\end{theorem}

\medskip

Theorem \ref{theo:03} boils up the application of the celebrated 
\emph{Kuroda-Birman Theorem} \cite[Theorem XI.9]{reed-simon-III}
which guarantees the existence and the completeness of the wave operators. In particular, the unitarity of $\bb{S}(V)$ is a consequence of the completeness of the wave operators.  

\medskip

A special class of bounded background potentials that meet the conditions of Theorem \ref{theo:02} and Theorem \ref{theo:03} are described in Remark \ref{rk-S-RC} and Remark \ref{rk-S-scaT}, respectively.

\medskip

It is worth to end this introductory section with few words about the strategy used for the proofs of the main results described above. 
Instead of working with the \virg{physical} operator $H_T$ we found more convenient to work with the unitarily equivalent (up to a scale factor) operator
\begin{equation}\label{eq:unit_map}
T\;:=\;\lambda^{-1}\;S_\lambda\; H_T\;S_\lambda^*
\end{equation}
obtained from $H_T$ via the unitary shift 
\begin{equation}\label{eq:intro_transl}
(S_\lambda\psi)(x)\;:=\;\psi\left(x-\frac{1}{\lambda}\right)\;=\;\psi(x+x_c)\;,\qquad\quad \psi\in L^2(\R)\;.
\end{equation}
The advantage relies on the fact that $T$ has a simpler expression with respect to $H_T$. In fact, at least formally, one has that $T=pxp$. It turns out that Theorem \ref{theo:01}, Theorem \ref{theo:02}  and Theorem \ref{theo:03} are nothing more that the transposition via the conjugation by $S_\lambda$ of the equivalent results proved for $T$ in Proposition \ref{prop:self-T}, Proposition \ref{theo:hs} and Proposition \ref{theo:waveop}, respectively. Anyway, the passage from the results concerning $T$ to the related results concerning $H_T$ is described in some detail in Remark \ref{rk-S-self}, Remark \ref{rk-S-RC} and Remark \ref{rk-S-scaT}.

%-----------------%
\medskip
\noindent
{\bf Structure of the paper.} In
{\bf Section \ref{sect:sob_ineq}} we recall some basic result  for the operator $T$ originally  obtained in \cite{denittis-lenz-20}  and we provide  new results about the regularity and the decay of the elements of the domain of $T$.  
{\bf Section \ref{sec_pert}} contains the results about the self-adjoint and relatively compact perturbations of the operator $T$.
Finally, {\bf Section \ref{sect_scatt}} provides the results about the scattering theory.

%-----------------%
\medskip
\noindent
{\bf Acknowledgements.} 
GD's research is supported
 by
the  grant \emph{Fondecyt Regular} -  1190204.
The authors are indebted to 
Olivier Bourget, Claudio Fernandez,
Marius Mantoiu and Serge Richard for many stimulating discussions.

%-----------------%
\section{Analysis of the domain}\label{sect:sob_ineq}
In this section we will provide some result about regularity and decay properties for for element in the domain of the operator $T$. Such  results can be inmediately transported to elements in the domain of $H_T$ in view of the unitary mapping \eqref{eq:unit_map}.

%-----%
\subsection{Basic facts about the unperturbed operator}\label{sec_fact_T}
We will start by recalling some important result concerning the spectral theory of the operator $T$ given by \eqref{eq:unit_map}. All the information presented here are taken form \cite{denittis-lenz-20}.

\medskip

An important role for the study of the operator $T$ is played by the bounded operator $B$ initially defined on elements  $\psi\in L^2(\R)\cap L^1(\R)$ by the integral formula
	 \begin{equation}\label{eq:int_ker_B_00}
	 (B\psi)(x)\;:=\;\int_\R\dd y\; \bb{B}\left(x,y\right)\psi(y)
	 \end{equation}
	 with kernel
	 \begin{equation}\label{eq:int_ker_B}
	 \bb{B}(x,y)\;:=\; \ii\frac{{\rm sgn}\left(x\right)-{\rm sgn}(y)}{2}\;J_0\left(2\sqrt{\left|xy\right|}\right)\;.
	 \end{equation}
	where  $J_0$ denotes the \emph{Bessel function} of the first kind \cite{gradshteyn-ryzhik-07}. We will use the  symbol $B$ to denote the  unique linear, bounded extension the dense defined operator \eqref{eq:int_ker_B_00}. It turns out that   $B$ is a unitary involution on $L^2(\R)$, \ie $B=B^*=B^{-1}$.

\medskip

Let $x$ be the  position operator on $L^2(\R)$, acting as multiplication by $x$ on its natural domain 
\begin{equation*}\label{eq:dom_pos}
\s{Q}(\R)\;:=\;\left\{\psi\in L^2(\R)\;\Big|\; \int_\R\dd x\; x^2|\psi(x)|^2\;<\;+\infty\right\}.
\end{equation*}
The involution $B$ introduced above intertwines between the operator $T$ and the position operators $x$. in fact it holds true that $T=-BxB$,  and as a consequence the domain of $T$ can be described a $\s{D}(T)=B[\s{Q}(\R)]$. 

\medskip

ative.\\

The resolvent of $T$ admits an explicit integral expression when evaluated on element of the dense domain $L^2(\R)\cap L^1(\R)$.
Let $z\in\C\setminus\R$ and consider the polar representation $z=|z|\expo{\pm\ii \phi}$ with  $0<\phi<\pi$. Let $R_z(T):=(T-z{\bf 1})^{-1}$ be the resolvent of $T$ at $z$. If $\psi\in L^2(\R)\cap L^1(\R)$, then it holds true that
\begin{equation}
	\big(R_z(T)\psi\big)(x)\;=\;\int_\R \dd y\;\big({\rm sgn}(x)+{\rm sgn}(y)\big)\; \bb{F}_z(x,y)\;\psi(y)
\end{equation}
Where 
\begin{equation}\label{eq:GOX_03}
	\begin{aligned}
		\bb{F}_z(x,y)\;:=&\;I_0\left(2\sqrt{|z|\min\{|x|,|y|\}}\expo{\pm \ii \left[\frac{\phi}{2}-\frac{\pi}{4}\big({\rm sgn}(x)+1\big)\right]}\right)\\
		&\times\;K_0\left(2\sqrt{|z|\max\{|x|,|y|\}}\expo{\pm  \ii\left[\frac{\phi}{2}-\frac{\pi}{4}\big({\rm sgn}(x)+1\big)\right]}\right)
	\end{aligned}
\end{equation}
with $I_0$ and $K_0$ are the modified Bessel functions of the firstand second kind, respectively.
In the special case  $z=\pm \ii$, by using \cite[eq. 10.61.1 \& eq. 10.61.2]{nist}, the formula above reduces to
\begin{equation*}
	\begin{aligned}
		\bb{F}_{\pm \ii}(x,y)\;:=&\;\lp{\rm ber}\left(2\sqrt{\min\{|x|,|y|\}}\right)\mp \ii\;{\rm sgn}(x){\rm bei}\left(2\sqrt{\min\{|x|,|y|\}}\right)\rp\\
		&\times\lp{\rm ker}\left(2\sqrt{\max\{|x|,|y|\}}\right)\mp \ii{\rm sgn}(x){\rm kei}\left(2\sqrt{\max\{|x|,|y|\}}\right)\rp, 
	\end{aligned}
\end{equation*}
with ber, bei, ker and kei  the  Kelvin functions of $0$-th order \cite[Chapter 55]{oldham-myland-spanier-09} or \cite[Section 10.61]{nist}. In accordance to \cite[Section 10.68]{nist}
let us introduce the following notation
$$
\begin{aligned}
M_0(s)\;&:=\;\sqrt{{\rm ber}(s)^2+{\rm bei}(s)^2}\\
N_0(s)\;&:=\;\sqrt{{\rm ker}(s)^2+{\rm kei}(s)^2}
\end{aligned}\;\;,\qquad s\geqslant0\;.
$$
This allows to rewrite the modulus of the kernel $\bb{F}_{\pm \ii}$ as follows 
\begin{equation}\label{eq:res_simp_002}
	\begin{aligned}
	|\bb{F}|(x,y)\;:&=\;	|\bb{F}_{\pm \ii}(x,y)|\\&=\;M_0\left(2\sqrt{\min\{|x|,|y|\}}\right)\;N_0\left(2\sqrt{\max\{|x|,|y|\}}\right)\;.
	\end{aligned}
\end{equation}
From \eqref{eq:res_simp_002} one infers immediately that  $|\bb{F}|$ is  invariant under the reflections
$x\mapsto-x$ and $y\mapsto-y$ and is independent on the sign in $\pm\ii$. For the study of the integrability properties of \eqref{eq:res_simp_002} we will make use of the 
inequalities
\begin{equation}\label{eq:asimMN}
\begin{aligned}
	M_0(s)^2\;&\leqslant\;C_M\frac{\expo{\sqrt{2}s}}{s}\;,\\
	N_0(s)^2\;&\leqslant\;C_N\frac{\expo{-\sqrt{2}s}}{s}\;,
\end{aligned}\qquad \forall\; s>0
\end{equation}
which can be deduced by the 
 asymptotic expansions  \cite[eq. 10.67.9 \& eq. 10.67.13]{nist}.
The exact value of the positive constants  $C_M$ and $C_N$
 is not important for the purposes of this work\footnote{The best constant $C_M$ is fixed by the maximum of the function $g_M(s):=sM_0(s)^2\expo{-\sqrt{2}s}$. A numerical inspection with
 {\tt Wolfram Mathematica (version 12.1)} shows that one can choose $C_M>\frac{1.65}{2\pi}$. A similar argument also provides $C_N>\frac{\pi}{2}$.}.

For small arguments ($s\sim0$) one has that the functions  $\mathrm{ber}$ and $\mathrm{bei}$ are continuous around the origin with  $\mathrm{ber}(0)=1$ and $\mathrm{bei}(0)=0$. Consequently,  also $M_0$ is continuous in $0$ and
\begin{equation*}
\lim_{s\to 0^+}M_0(s)\;=\;	M_0(0)\;=\;1\;.
\end{equation*}
On the other hand $\mathrm{kei}(0)=-\frac{\pi}{4}$, but $\mathrm{ker}$, and in turn $N_0$,  have a logarithmic pole in $0$. 
By using  the series expansions \cite[eq. 10.65.5]{nist} one can prove that
\begin{equation*}
	\lim_{s\to 0^+}\big(N_0(s)+\ln(s)\big)\;=\;\ln(2)-\gamma
	\end{equation*}
where  $\gamma$ is the Euler-Mascheroni constant.
For the next result, that will play a crucial role in Section \ref{sec_rel_com-pert}, we need to introduce the positive semi-axis $\R_+:=(0,+\infty)$ and the first quadrant $\R_+^2:=\R_+\times\R_+$. 
\begin{lemma}\label{lemma:square_kernel}
Let $w:\R_+\to\R_+\cup\{0\}$ be a non negative function
such that
$$
\int_{\R_+}\dd x\; \frac{w(x)}{\sqrt{x}}\;<\;+\infty 
$$
and  $|\bb{F}|$  the function defined by \eqref{eq:res_simp_002}. 
Then  the integral
$$
I_w\;:=\;\iint_{\R_+^2}\dd x\dd y\;w(x)\;|\bb{F}|(x,y)^2\;<\;+\infty\;,
$$
is finite.
\end{lemma}
\proof
Let us split the integral $I_w$ as follows
$$
I_w\;=\;I_{w}(\Sigma_1)\;+\;I_{w}(\Sigma_2)
$$
where we used the  notation
$$
I_{w}(\Sigma_i)\;:=\;\iint_{\Sigma_i}\dd x\dd y\;w(x)\;|\bb{F}|(x,y)^2\;,\qquad i=1,2
$$
with
$\Sigma_{1}:=\{(x,y)\in \R_+^2|x<y \}$ and $\Sigma_{2}:=\{(x,y)\in \R_+^2|y<x \}$. We start with the integral $I_{w}(\Sigma_1)$. 
In view of the definition  of $|\bb{F}|$ one has that 
$$
I_{w}(\Sigma_1)\;:=\;\iint_{\Sigma_{1}}\dd x\dd y\;w(x)\;M_0\left(2\sqrt{x}\right)^2\;N_0\left(2\sqrt{y}\right)^2\;.
$$
From the inequalities \eqref{eq:asimMN} one infers  that
$$
M_0\left(2\sqrt{x}\right)^2\;N_0\left(2\sqrt{y}\right)^2\;\leqslant\;C\; \frac{\expo{2\sqrt{2x}}\expo{-2\sqrt{2y}}}{\sqrt{xy}}\;,\qquad \forall\; (x,y)\in \Sigma_{1}\;.
$$
with $C:=\frac{C_MC_N}{4}>0$.
As a consequence one has that
$$
\begin{aligned}
I_{w}(\Sigma_{1})\;&\leqslant\;C\; \iint_{\Sigma_{1}}\dd x\dd y\;w(x)\;\frac{\expo{2\sqrt{2x}}\expo{-2\sqrt{2y}}}{\sqrt{xy}}\\
&=\;C\; \int_{\R_+}\dd x\; w(x)\;  \frac{\expo{2\sqrt{2x}}}{\sqrt{x}}\left(\int_{x}^{+\infty}\dd y\;\frac{\expo{-2\sqrt{2y}}}{\sqrt{y}}\right)\\
&=\;\frac{C}{\sqrt{2}}\; \int_{\R_+}\dd x\; \frac{w(x)}{\sqrt{x}}\;<\;+\infty \\
\end{aligned}
$$
where the first equality (second line) follows from the Tonelli's theorem and the 
last inequality is guaranteed by  hypothesis. The treatment of the integral $I_{w}(\Sigma_{2})$ is quite similar. Indeed, one has that
$$
\begin{aligned}
I_{w}(\Sigma_2)\;&=\;\iint_{\Sigma_{2}}\dd x\dd y\;w(x)\;N_0\left(2\sqrt{x}\right)^2\;M_0\left(2\sqrt{y}\right)^2\\
&\leqslant\;C\; \int_{\R_+}\dd x\; w(x)\;  \frac{\expo{-2\sqrt{2x}}}{\sqrt{x}}\left(\int^{x}_{0}\dd y\;\frac{\expo{2\sqrt{2y}}}{\sqrt{y}}\right)\\
&=\;\frac{C}{\sqrt{2}}\; \int_{\R_+}\dd x\; w(x)\;  \frac{1-\expo{-2\sqrt{2x}}}{\sqrt{x}}\\
&\leqslant\;\frac{C}{\sqrt{2}}\;\int_{\R_+}\dd x\; \frac{w(x)}{\sqrt{x}}\;<\;+\infty\;.
\end{aligned}
$$
This concludes the proof.
\qed

%---%

\subsection{Regularity and decay}
This section is devoted to the description of elements in  $\s{D}(T)$. Let us start from a preliminary result.
\begin{lemma}\label{lemma:_dom_1}
For every $\psi\in \s{D}(T)$ it holds true that
	\begin{equation}\label{eq:boun_01}
		\lv B\psi\rv_{L^1}^2\;\leqslant\; 2\pi\; \lv T\psi\rv_{L^2}\;\lv \psi\rv_{L^2}\;.
	\end{equation}
As a consequence one has that $B[\s{D}(T)]\subset L^1(\R)$. \end{lemma}
\begin{proof}
	 Let $v>0$ be an arbitrary  number and consider the identity
	$$
		\lv B\psi\rv_{L^1}\;=\; \int_{-\infty}^{+\infty}\dd x \abs{(B\psi)(x)} \left(x^2+v^2\right)^{\frac{1}{2}} \left(x^2+v^2\right)^{-\frac{1}{2}}\;.
		$$
By the classical Cauchy-Bunyakovsky-Schwarz inequality one gets
$$
\lv B\psi\rv_{L^1}^2\;\leqslant\;\lp \int_{-\infty}^{+\infty}\dd x\; \abs{(B\psi)(x)}^2 \left(x^2+v^2\right)\rp\lp \int_{-\infty}^{+\infty}\dd x\;  \left(x^2+v^2\right)^{-1}\rp\;.
$$
By recalling the equivalence between $T$ and the position operator described in Section \ref{sec_fact_T} one can rewrite the first integral as follows:
$$
\begin{aligned}
\int_{-\infty}^{+\infty}\dd x\; \abs{(B\psi)(x)}^2 \left(x^2+v^2\right)\;=\;\|BT\psi\|^2_{L^2}\;+\; v^2\;\|B\psi\|^2_{L^2}\;.
\end{aligned}
$$
By using the fact that $B$ is unitary and the known formula
$$
\int_{-\infty}^{+\infty}\dd x\;  \left(x^2+v^2\right)^{-1}\;=\;\frac{\pi}{v}
$$
one gets
$$
\begin{aligned}
\lv B\psi\rv_{L^1}^2\;&\leqslant\;\frac{\pi}{v}\left(\|T\psi\|^2_{L^2}+ v^2\;\|\psi\|^2_{L^2}\right)\\
\end{aligned}
$$
independently of $v>0$. By minimizing with respect to $v$ one obtains the bound \eqref{eq:boun_01}.
\end{proof}

\medskip

The following lemma makes use of the explicit form of the kernel of $B$.
\begin{lemma}[Boundedness]\label{lemma:L_inf_dom}
	For every $\psi\in \s{D}(T)$ it holds true that
	\begin{equation}\label{eq:boun_01-01}
		\lv \psi\rv_{L^\infty}\;\leqslant\; \left(2\pi\; \lv T\psi\rv_{L^2}\;\lv \psi\rv_{L^2}\right)^{\frac{1}{2}}
		\;\leqslant\; \sqrt{\pi}\; \left(\lv T\psi\rv_{L^2}\;+\;\lv \psi\rv_{L^2}\right)
		\;.
	\end{equation}
As a consequence one has that $\s{D}(T)\subset L^\infty(\R)$.
\end{lemma}
\proof
Since $B$ is an involution one has that $\psi=B(B\psi)$, and in view of Lemma \ref{lemma:_dom_1} one infers that $B\psi\in L^1(\R)\cap L^2(\R)$. Therefore, the action of $B$ on $B\psi$ can be computed via the integral formula \eqref{eq:int_ker_B_00} and \eqref{eq:int_ker_B}. This provides
$$
\begin{aligned}
|\psi(x)|\;&=\;\left|\int_\R\dd y\; \bb{B}\left(x,y\right)(B\psi)(y)\right|\\
&\leqslant\;\int_\R\dd y\; \left|\bb{B}\left(x,y\right)\right|\;\left|(B\psi)(y)\right|
\\
&\leqslant\;\int_\R\dd y\; \left|(B\psi)(y)\right|\;=\;\lv B\psi\rv_{L^1}
\end{aligned}
$$	
where we used the bound $\left|\bb{B}\left(x,y\right)\right|\leqslant 1$. Since the inequality above holds for every $x\in\R$ one gets $\lv \psi\rv_{L^\infty}\leqslant \lv B\psi\rv_{L^1}$.
The rest of the proof follows from the inequality in Lemma \ref{lemma:_dom_1}.
\qed

\medskip

The next result describes the continuity properties of elements of 
$\s{D}(T)$.

\begin{lemma}[H\"older continuity]\label{lemma_holder}
	Let $\psi\in \s{D}(T)$ and  $x,y\in \R\setminus\{0\}$ with ${\rm sgn}(x)={\rm sgn}(y)$. Then, it holds true that

	\begin{equation}\label{eq:main_ineq}
		\abs{\psi(x)-\psi(y)}^2\;\leqslant\; G_k\;\abs{x-y}^k\;\norm{T\psi}_{L^2}^{1+k}\;\norm{\psi}_{L^2}^{1-k}
	\end{equation}
	for every 	$0\leqslant k<1$,
with $G_k$ a constant depending only on $k$. As a consequence the element of $\s{D}(T)$ are $\alpha$-H\"older continuous in $\R\setminus\{0\}$ with $0\leqslant \alpha<\frac{1}{2}$.
\end{lemma}
\proof
By using the identity $\psi=B(B\psi)$ and the integral expression of $B$ as in the proof of Lemma \ref{lemma:L_inf_dom} one obtains
$$
\begin{aligned}
\abs{\psi(x)-\psi(y)}\;&=\;\abs{\int_\R\dd s\;\left[\bb{B}\left(x,s\right)-\bb{B}\left(y,s\right)\right]\;(B\psi)(s)}\\
&\leqslant\;\int_\R\dd s\;\abs{J_0\left(2\sqrt{\left|sx\right|}\right)-J_0\left(2\sqrt{\left|sy\right|}\right)}\;\abs{(B\psi)(s)}
\end{aligned}
$$
where in the last inequality the hypothesis ${\rm sgn}(x)={\rm sgn}(y)$ has been used. In view of the integral representation of the Bessel $J_0$ 
	\begin{equation}\label{int_rep}
		J_0(z)=\frac{1}{2\pi} \int_{-\pi}^\pi\dd\tau\; \expo{\ii z \sin( \tau)}
	\end{equation}
one obtains
\begin{equation}\label{eq:ineq:0X0}
\abs{\psi(x)-\psi(y)}\;\leqslant\;\frac{1}{2\pi}\int_\R\dd s \int_{-\pi}^\pi\dd\tau\;g_{(x,y)}(s,\tau)\;\abs{(B\psi)(s)}
\end{equation}
where
$$
g_{(x,y)}(s,\tau)\;:=\;\abs{\expo{\ii 2\sqrt{\left|sx\right|} \sin( \tau)}-\expo{\ii 2\sqrt{\left|sy\right|} \sin( \tau)}}\;.
$$
For fixed values of $s$, and using the classical estimate $|1-\expo{\ii s}|\leqslant |s|$ one gets
$$
\begin{aligned}
g_{(x,y)}(s,\tau)\;&\leqslant\;\min\left\{2,2\abs{\sqrt{\left|sx\right|}-\sqrt{\left|sy\right|}}\abs{\sin(\tau)}\right\}\;.\end{aligned}
$$
Since the minimum between two numbers is dominated by any weighted geometric mean of the same one obtains that
$$
\begin{aligned}
g_{(x,y)}(s,\tau)\;&\leqslant\;2^{1-k}\left(2\abs{\sqrt{\left|sx\right|}-\sqrt{\left|sy\right|}}\abs{\sin(\tau)}\right)^k\;,\qquad \forall\; 0\leqslant k\leqslant 1\;.
\end{aligned}
$$
After inserting the last inequality in \eqref{eq:ineq:0X0} one gets
\begin{equation}\label{eq:ineq:0X1}
\abs{\psi(x)-\psi(y)}\;\leqslant\;{C_k}\int_\R\dd s \;\abs{\sqrt{\left|sx\right|}-\sqrt{\left|sy\right|}}^k\;\abs{(B\psi)(s)}
\end{equation}
where
\begin{equation}
	C_k\;:=\;\frac{1}{\pi}	\int_{-\pi}^\pi\dd\tau\abs{\sin(\tau)}^k\;=\;\frac{4}{\pi}\int_{0}^{\frac{\pi}{2}}\dd\tau\;\sin(\tau)^k\;=\;\frac{2}{\sqrt{\pi}}\;\frac{\Gamma\left(\frac{k+1}{2}\right)}{\Gamma\left(\frac{k+2}{2}\right)}
	\end{equation}
with $\Gamma$ denoting the gamma function.	Observing that
$$
\begin{aligned}
\abs{\sqrt{\left|sx\right|}-\sqrt{\left|sy\right|}}^k\;&=\;|s|^{\frac{k}{2}}\;\abs{\sqrt{\left|x\right|}-\sqrt{\left|y\right|}}^k\\
&\leqslant\;|s|^{\frac{k}{2}}\;\big|{\left|x\right|}-{\left|y\right|}\big|^\frac{k}{2}\;\leqslant\;|s|^{\frac{k}{2}}\;|{x}-{y}|^\frac{k}{2}
\end{aligned}
$$	
one obtains
\begin{equation}\label{eq:ineq:0X2}
\abs{\psi(x)-\psi(y)}\;\leqslant\;{C_k}\;|x-y|^\frac{k}{2}\int_\R\dd s \;|s|^{\frac{k}{2}}\;\abs{(B\psi)(s)}\;.
\end{equation}
By inserting inside the integral the identity
$(s^2+v^2)^{\frac{1}{2}} (s^2+v^2)^{-\frac{1}{2}}$ with $v>0$, and using 
the same trick  employed in the proof of Lemma \ref{lemma:_dom_1} 
one gets	
\begin{equation}\label{eq:ineq:0X3}
\abs{\psi(x)-\psi(y)}^2\;\leqslant\;C_k^2\;|x-y|^k\;q_k(v)\;\left(\|T\psi\|^2_{L^2}+ v^2\;\|\psi\|^2_{L^2}\right)
\end{equation}
where
$$
q_k(v)\;:=\;\int_\R\dd s \;\frac{|s|^{k}}{s^2+v^2}
\;=\;\frac{1}{v^{1-k}}\frac{\pi}{\cos\left(\frac{k\pi}{2}\right)}\;,\qquad 0\leqslant k <1\;.
$$
A minimization procedure on $v$ provides
$$
\min_{v>0}\frac{\|T\psi\|^2_{L^2}+ v^2\;\|\psi\|^2_{L^2}}{v^{1-k}}\;=\;I_k\;\|T\psi\|^{1+k}_{L^2}\;\|\psi\|^{1-k}_{L^2}
$$	
with
$$
I_k\;:=\;\frac{2}{(1+k)^{\frac{1+k}{2}}(1-k)^{\frac{1-k}{2}}}\;.
$$	
Since inequality \eqref{eq:ineq:0X3} holds for every $v>0$, it holds true also when $v$ coincides with the minimizer. This provides the inequality \eqref{eq:main_ineq} with the constant given by $G_k:=C^2_k I_k\frac{\pi}{\cos\left(\frac{k\pi}{2}\right)}$. This concludes the proof.
\qed

\medskip

The last result allows to control the behavior of   the elements of the domain $\s{D}(T)$ around the critical point $x=0$ and gives an upper bound for the decay rate at infinity .
\begin{lemma}[Global control of the behavior]\label{lemma:fdecay}
Let $\psi\in \s{D}(T)$. Then, it holds true that for some postive constant $K$,
	\begin{equation}
		\abs{\psi(x)}\;\leqslant\;K\abs{x}^{-\frac{1}{4}}\;\|T\psi\|^{\frac{1}{4}}_{L^2}\;\|\psi\|^\frac{3}{4}_{L^2}\;,\quad\forall\; x\in\R\setminus\{0\}\;.
	\end{equation}
\end{lemma}
\proof
	Let us start by observing that from  \cite[eq. 10.7.8]{nist}
one infers that there exists a 	positive constant $M>0$ such that
	\[
|J_0(z)|\;\leqslant\;\frac{M}{\sqrt{z}}\;,\qquad \forall\;z\geqslant0\;.
\]
The exact value of the  constants  $M$
 is not important for the purposes of this work\footnote{The best constant $M$ is fixed by the maximum of the function $f(z):=|J_0(z)|\sqrt{z}$. A numerical inspection with
 {\tt Wolfram Mathematica (version 12.1)} shows that one can choose $M>\sqrt{\frac{2}{\pi}}$.}. 	
 Using the same argument in the proof of Lemma \ref{lemma_holder} one gets	
$$
\begin{aligned}
\abs{\psi(x)}\;&\leqslant\;\int_\R\dd s\;\abs{J_0\left(2\sqrt{\left|sx\right|}\right)}\;\abs{(B\psi)(s)}\\
\;&\leqslant\;M \abs{x}^{-\frac{1}{4}}\;\int_\R\dd s\;\abs{s}^{-\frac{1}{4}}\;\abs{(B\psi)(s)}\;.
\end{aligned}
$$
which shows that $\abs{\psi(x)}$ is dominated by $\abs{x}^{-\frac{1}{4}}$. For the determination of the constant one can 
inserting inside the integral the identity $(s^2+v^2)^{\frac{1}{2}} (s^2+v^2)^{-\frac{1}{2}}$ with $v>0$ and, with
the same trick  employed in the proof of Lemma \ref{lemma:_dom_1}, 
one gets	
$$
\begin{aligned}
\int_\R\dd s\;\abs{s}^{-\frac{1}{4}}\;\abs{(B\psi)(s)}\;&\leqslant\;\left(\int_{-\infty}^{+\infty}\dd s\;\frac{\abs{s}^{-\frac{1}{2}}}{\left(s^2+v^2\right)}\right)^{\frac{1}{2}}\left(\|T\psi\|^2_{L^2}+ v^2\;\|\psi\|^2_{L^2}\right)^{\frac{1}{2}}\\
&\;=\;\frac{2^{\frac{1}{4}}\sqrt{\pi}}{v^{\frac{3}{4}}}\left(\|T\psi\|^2_{L^2}+ v^2\;\|\psi\|^2_{L^2}\right)^{\frac{1}{2}}\;.
\end{aligned}
$$
After minimizing the last inequality with respect  $v>0$, and grouping all constants in $K>0$, one finally obtains
$$
\abs{\psi(x)}\;\leqslant\;K\;\abs{x}^{-\frac{1}{4}}\;\|T\psi\|^{\frac{1}{4}}_{L^2}\;\|\psi\|^\frac{3}{4}_{L^2}\;.
$$
This concludes the proof. \qed

\begin{remark}[Synopsis]
Let us summarize the main properties of elements in the domain $\s{D}(T)$ obtained in this section. Lemma \ref{lemma:L_inf_dom} provides:
\begin{itemize}
\item[(i)] $\s{D}(T)\in L^2(\R)\cap L^\infty(\R)$\;.
\end{itemize}
By combining Lemma \ref{lemma:L_inf_dom} and Lemma \ref{lemma_holder} one infers that: 
\begin{itemize}
\item[(ii)] Every $\psi\in\s{D}(T)$ is at least H\"older continuous in $\R\setminus\{0\}$ of order $\alpha<\frac{1}{2}$. Moreover, the possible discontinuity in the critical point $x=0$ is of the \emph{first kind} meaning that both limits $\lim_{x\to0^\pm}\psi(x)=L_\pm$ exist and are finite.
\end{itemize}
Finally Lemma \ref{lemma:fdecay} tells us that: 
\begin{itemize}
\item[(iii)] The global behavior of $\psi\in\s{D}(T)$ in $\R\setminus\{0\}$  is dominated by $|x|^{-\frac{1}{4}}$. In particular the elements of $\s{D}(T)$ vanish at infinity.
\end{itemize}
Let us point out that the existence of singular discontinuous elements in $\s{D}(T)$, as stipulated by (ii) is unavoidable. Indeed, we know from \cite{denittis-lenz-20} that smooth functions aren't enough to provide a core for $T$, while a core is provided by  $\s{S}(\R)+\C[\kappa_0]$. The function
$\kappa_0(x)\propto {\rm sgn }(x){\rm kei}(2\sqrt{|x|})$
is smooth in $\R\setminus\{0\}$ and has a discontinuity of first kind in $x=0$. The decay at infinity of $\kappa_0$ us quite rapid. Indeed on has that $\kappa_0(x)={\rm o}(\expo{-\sqrt{|x|}})$. Finally, it should be noted that Lemma \ref{lemma:fdecay} provides a somehow \emph{weak} information about the decay of elements of $\s{D}(T)$ at infinity. In fact a decay of type $|x|^{-\frac{1}{4}}$  it is not enough to guarantee that $\psi\in L^2(\R)$. 
However, Lemma \ref{lemma:fdecay} gives an important information 
that will be crucial in Proposition \ref{prop:self-T}.
Let $\psi\in \s{D}(T)$ and define $\phi(x):=|x|^{\frac{1}{4}}\psi(x)$. Then $\|\phi\|_{L^\infty}\leqslant \infty$. Moreover, from the proof of  Lemma \ref{lemma:fdecay} one also infers that 
\begin{equation}\label{eq:self-02-rk}
\|\phi\|_{L^\infty}^2\;\leqslant\;\frac{C}{v^{\frac{3}{2}}}\left(\|T\psi\|^2_{L^2}+ v^2\;\|\psi\|^2_{L^2}\right)
\end{equation}
with $C>0$ a suitable constant and for every $v>0$.
Let us conclude this remark by observing that the properties (i), (ii) and (iii)
remain valid for elements of the domain $\s{D}_{T}$ of the thermal Hamiltonian $H_T$ with the only difference that the critical point has to be shifted from $0$ to
 $x_c$ as effect of the translation $S_\lambda$.
   \hfill $\blacktriangleleft$
\end{remark}

%--------------%
\section{Perturbations by potential}\label{sec_pert}
In this section we will study the perturbations of the operator $T$ by multiplicative, real-valued potentials.

%--------------%
\subsection{Self-adjoint  perturbations}\label{sec_self-pert}
Let  $W:\R \to \R$ be a  real-valued function. With a slight abuse of notation we will denote with the same symbol also the multiplication operator defined on $\psi\in L^2(\R)$ by $(W\psi)(x):=W(x)\psi(x)$. This is a self-adjoint operator with natural domain given by
$$
\s{D}(W)\;:=\;\{\psi\in L^2(\R)\;|\; W\psi\in L^2(\R)\}\;.
$$
The next result provides condition on the potential $W$ for the self-adjointness of $T+W$.
\begin{proposition}[Self-adjoint perturbations]\label{prop:self-T}
	Let $W:\R \to \R$ be  a real-valued function such that 
	$$
	W(x)\;:=\;|x|^{\frac{1}{4}}\;V_1(x)\,+\; V_2(x)\;,\qquad x\in\R
	$$
with $V_1\in L^2(\R)$ and $V_2\in L^\infty(\R)$. 	
Then, for every $\epsilon>0$ there exists a positive constant $B_\epsilon>0$ such that
	\begin{equation}\label{ineq_kato}
		\norm{W\psi}_{L^2}^2\;\leqslant\;\epsilon\norm{T\psi}_{L^2}^2\;+\;B_\epsilon\norm{\psi}_{L^2}^2\;.
	\end{equation}
As a consequence $T+W$ defines a self-adjoint operator with domain $\s{D}(T)$.
\end{proposition}
\proof
Let $\psi\in\s{D}(T)$. Then one has that
\begin{equation}\label{eq:self-01}
\begin{aligned}
\|W\psi\|_{L^2}\;&\leqslant\;\|V_1\phi\|_{L^2}\;+\:\|V_2\psi\|_{L^2}\\
&\leqslant\;\|V_1\|_{L^2}\;\|\phi\|_{L^\infty}\;+\:\|V_2\|_{L^\infty}\;\|\psi\|_{L^2}
\end{aligned}
\end{equation}
where $\phi(x):=|x|^{\frac{1}{4}}\psi(x)$ with $\|\phi\|_{L^\infty}<\infty$ in view of Lemma \ref{lemma:fdecay}. This proves that 
$\s{D}(T)\subseteq \s{D}(W)$.  By combining \eqref{eq:self-01} 
with the inequality \eqref{eq:self-02-rk} (also deduced from Lemma \ref{lemma:fdecay})
 one obtains
\begin{equation}\label{eq:self-03}
\begin{aligned}
\|W\psi\|_{L^2}^2\;&\leqslant\;\left(\|V_1\|_{L^2}\;\|\phi\|_{L^\infty}\;+\:\|V_2\|_{L^\infty}\;\|\psi\|_{L^2}\right)^2\\
&\leqslant\;2\left(\|V_1\|_{L^2}^2\;\|\phi\|_{L^\infty}^2\;+\:\|V_2\|_{L^\infty}^2\;\|\psi\|_{L^2}^2\right)
\\
&\leqslant\;\frac{2C\; \|V_1\|_{L^2}^2}{v^{\frac{3}{2}}}\; \|T\psi\|^2_{L^2}\;+\;\left(2Cv^{\frac{1}{4}}+\|V_2\|_{L^\infty}^2\right)\;\|\psi\|_{L^2}^2\;.
\end{aligned}
\end{equation}
Since \eqref{eq:self-03} holds for every $v>0$ one can  always find  a $v_\epsilon$ such that $2C\|V_1\|_{L^2}^2v_\epsilon^{-\frac{3}{2}}=\epsilon$. By setting $B_\epsilon:=2Cv_\epsilon^{\frac{1}{4}}+\|V_2\|_{L^\infty}^2$ one obtains the inequality \eqref{ineq_kato}.
The latter implies that $W$ is infinitesimally small with respect to $T$ in the sense of Kato (\cf \cite[eq. X.19a \& eq. X.19b]{reed-simon-II}). Consequently $T+W$ defined a self-adjoint operator with domain $\s{D}(T)$ in view of the Kato-Rellich theorem \cite[Theorem X.12]{reed-simon-II}.
\qed

\begin{remark}[Self-adjoint perturbation of the thermal Hamiltonian]\label{rk-S-self}
Let $W$ a potential which  meets the conditions of Proposition \ref{prop:self-T}. Then $T+W$ is a self-adjoint operator with domain $\s{D}(T)$. As a consequence $H_T+\widetilde{W}$, with
$$
\widetilde{W}\;:=\;\lambda\;S^*_\lambda W S_\lambda
$$
 is a self-adjoint perturbation of the thermal Hamiltonian $H_T$
in view of the equivalence \eqref{eq:unit_map}. The splitting of $W$ stipulated in Proposition \ref{prop:self-T} translates to
$$
\widetilde{W}(x)\;=\; |x-x_c|^{\frac{1}{4}}\;\widetilde{V}_1(x)\,+\; \widetilde{V}_2(x)
$$
where $\widetilde{V}_i:=\lambda S^*_\lambda V_i S_\lambda$ and $i=1,2$. Since the conjugation by $S_\lambda$ is a translation one obtains that $\widetilde{V}_1\in L^2(\R)$ and $\widetilde{V}_2\in L^\infty(\R)$. The considerations above provide the main argument to deduce  Theorem \ref{theo:01} from Proposition \ref{prop:self-T}.
   \hfill $\blacktriangleleft$
\end{remark}

%--------------%
\subsection{Relatively compact perturbations}\label{sec_rel_com-pert}

The task of this section is to find  suitable class of 
\emph{relatively compact} (or
$T$-compact) perturbations $W:\R\to\R$ of the free operator $T$. 
This means that $\s{D}(T)\subseteq \s{D}(W)$ and the product $WR_{\pm \ii}(T)$  must be a  compact operator \cite[Section XIII.4]{reed-simon-IV}, with 
$$
R_{z}(T)\;:=\;(T-z{\bf 1})^{-1}\;,\qquad z\in\C\setminus\R
$$
the resolvent of $T$ evaluated at $z$.
In the next result we will provide  conditions on $W:\R\to\R$ such that the product $WR_{\pm \ii}(T)$ is a \emph{Hilbert-Schmidt} operator.
\begin{proposition}\label{theo:hs}
	Let $W:\R\to\R$ be a real-valued function such that
	$$
	W(x)\;=\;|x|^{\frac{1}{4}} V(x)
	$$
	with $V\in L^2(\R)$.
	Then $WR_{\pm \ii}(T)$ are Hilbert-Schmidt operators (hence compact).
\end{proposition}
\proof
The resolvents 	$R_{\pm \ii}(T)$ are integral operators with explicit integral kernels given by \eqref{eq:GOX_03}
when evaluated on the dense domain $L^1(\R)\cap L^2(\R)$.
Let us introduce the symbol 
\begin{equation}\label{eq.Kern-K}
\bb{K}_\pm(x,y)\;:=\;\big({\rm sgn}(x)+{\rm sgn}(y)\big)\;W(x)\;\bb{F}_{\pm \ii}(x,y)
\end{equation}
for the integral kernel of $WR_{\pm \ii}(T)$. Since
 the term ${\rm sgn}(x)+{\rm sgn}(y)$ vanishes when $x$ and $y$ have different signs,  one has that
	\begin{equation*}
		\iint_{\R^2}\dd x\dd y\;|\bb{K}_\pm(x,y)|^2\;=\;2\;(I_+\;+\;I_-)
	\end{equation*}
where	
$$
\begin{aligned}
I_{\pm}\;:&=\;\iint_{\R_\pm^2}\dd x\dd y\;|W(x)|^2\;|\bb{F}|(x,y)^2\\
&=\;\iint_{\R_+^2}\dd x\dd y\;|W(\pm x)|^2\;|\bb{F}|(x,y)^2\;.
\end{aligned}
$$	
In the definition of the integrals $I_{\pm}$ we used the following notation for the domains of integration: $\R_\pm:=\{x\in\R|\pm x>0\}$ are the positive and negative semi-axis
and 
$\R_\pm^2:=\R_\pm \times \R_\pm$	are the first and third quadrants. The function $|\bb{F}|$ is defined by 
\eqref{eq:res_simp_002}. In the second equality the invariance property of $|\bb{F}|$ under the reflections
$x\mapsto-x$ and $y\mapsto-y$ has been exploited. 
In view of the splitting above, and after a comparison with the notation introduced in Lemma \ref{lemma:square_kernel} one obtains that
$$
I_\pm\;=\;I_{w_\pm}\;:=\;\iint_{\R_+^2}\dd x\dd y\;w_\pm(x)\;|\bb{F}|(x,y)^2
$$
where the functions $w_\pm:\R_+\to\R_+\cup\{0\}$ are defined as follow:
$$
\begin{aligned}
w_+(x)\;&:=\;|W(x)|^2\;, \qquad
w_-(x)\;&:=\;|W(-x)|^2\;.
\end{aligned}
$$
Since
$$
\int_{\R_+}\dd x\; \frac{w_\pm(x)}{\sqrt{x}}\;=\;\int_{\R_\pm}\dd x\; |V(x)|^2\;\leqslant\;\|V\|_{L^2}
$$
one can apply Lemma \ref{lemma:square_kernel} to conclude that 
 $I_\pm<+\infty$. This implies that the kernels \eqref{eq.Kern-K} are square-integrable, \ie  $\bb{K}_\pm\in L^2(\R^2)$.
Therefore, the kernels $\bb{K}_\pm$ define two Hilbert-Schmidt operators $K_\pm$ on $L^2(\R)$ \cite[Theorem VI.23]{reed-simon-I}. Moreover, one has that $K_\pm$ coincides with 
$WR_{\pm \ii}(T)$ on the dense domain $L^1(\R)\cap L^2(\R)$.
This is enough to conclude that $K_\pm= WR_{\pm \ii}(T)$ and this completes the proof.
\qed

\medskip

The hypotheses on $W$ stipulated in 
 Proposition \ref{theo:hs} are stronger than the hypotheses in Proposition \ref{prop:self-T}. Consequently it holds true that under the hypotheses of Proposition \ref{theo:hs} $T+W$ is a self-adjoint operator and $W$ is a relatively compact perturbation of $T$. As a consequence of the celebrated  Weyl Theorem  \cite[Theorem XIII.14]{reed-simon-IV} one obtains the equality of the  essential spectra, \ie
$$
\sigma_{\rm ess}(T+W)\;=\;\sigma_{\rm ess}(T)\;=\;\R\;.
$$
It is also worth noting that under the  hypotheses of Proposition \ref{theo:hs} it holds true that $WR_{z}(T)$ is a
Hilbert-Schmidt operator for every $z\in\C\setminus\R$. This follows from the identity
$$
WR_{z}(T)\;=\;WR_{\pm\ii}(T)\left[(T\mp \ii {\bf 1})R_{z}(T)\right]\;,
$$
the boundedness of the operator inside the square brackets and the fact that the class of Hilbert-Schmidt operators is a two-sided ideals inside the bounded operator.

\medskip

The next result provides a class of prototipe bounded potentials which meet the condition stipulated in Proposition \ref{theo:hs}.
\begin{corollary}\label{cor:class_HS1}
The potentials 
$$
W_r(x)\;:=\;\left(1+x^2\right)^{-r}\;,\qquad r>\frac{1}{8}
$$ 
are  relatively compact perturbations of the operator $T$ and
$W_rR_{z}(T)$ are Hilbert-Schmidt operators for every $z\in\C\setminus\R$.
\end{corollary}
\proof
An explicit computation provides
$$
\int_\R\dd x\; \frac{|W_r(x)|^2}{|x|^{\frac{1}{2}}}\;=\;4\Gamma\left(\frac{5}{4}\right)\frac{\Gamma\left(2r-\frac{1}{4}\right)}{\Gamma\left(2r\right)}\;, \qquad r>\frac{1}{8}
$$
where $\Gamma$ denotes the Gamma function. Therefore the potentials $W_r$, with $r>\frac{1}{8}$, satisfy the hypotheses of Proposition \ref{theo:hs}. 
\qed

\medskip

For the next result let us introduce the \emph{Japanese brackets}
$\langle x \rangle:\R\to\R_+$ defined by
$$
\langle x \rangle(x)\;:=\;\left(1+x^2\right)^{\frac{1}{2}}\;.
$$
\begin{corollary}\label{cor:<x>s}
Let $W$ be a real valued potential such that
$$
W\;\langle x \rangle^s\; \in\; L^\infty(\R)
$$
for some $s>\frac{1}{4}$. 
Then $W$
is a  relatively compact perturbation  of the operator $T$ and
$WR_{z}(T)$ are Hilbert-Schmidt operators for every $z\in\C\setminus\R$.
\end{corollary}
\proof
Consider the identity
$$
WR_{z}(T)\;=\;\big(W\;\langle x \rangle^s\big)\big(\langle x \rangle^{-s} R_{z}(T)\big)\;.
$$
Since $\langle x \rangle^{-s}=W_{\frac{s}{2}}$ one gets from Corollary \ref{cor:class_HS1} that $\langle x \rangle^{-s} R_{z}(T)$ is Hilbert-Schmidt whenever $s>\frac{1}{4}$.
To conclude the proof it is enough to observe that $W\langle x \rangle^s$ is a bounded operator according to the hypothesis and that the class of the Hilbert-Schmidt operators is an ideal inside the bounded operators.
\qed

\begin{remark}[Relatively compact perturbations of the thermal Hamiltonian]\label{rk-S-RC}
In the spirit of Remark \ref{rk-S-self} one can translate the conditions for the relative compactness of a perturbation from the operator $T$ to the thermal Hamiltonian $H_T$ just by using the 
equivalence \eqref{eq:unit_map} implemented by the translation 
$S_\lambda$. As a result one gets that potentials of the type
$$
\widetilde{W}(x)\;=\; |x-x_c|^{\frac{1}{4}}\;\widetilde{V}(x)
$$
with 
$\widetilde{V}\in L^2(\R)$ are relatively compact perturbations of 
$H_T$ of Hilbert-Schmidt type
 This is the key observation    to deduce  Theorem \ref{theo:02} from Proposition \ref{theo:hs}. It is worth to translate the result  of Corollary  \ref{cor:<x>s} in terms of the 
 background potential $V$ which defines the perturbation of the 
 thermal Hamiltonian via the equation \eqref{eq:intro_05}. By using the usual argument involving the shift operator $S_\lambda$ one infers that the condition
\begin{equation}\label{con_alep1}
 V\;\langle x-x_c \rangle^{s+1}\; \in\; L^\infty(\R)\;,
\end{equation}
 for some $s>\frac{1}{4}$, implies that $W_V$ is a relatively compact perturbation of 
$H_T$ of Hilbert-Schmidt type. In the last equation we introduced the shifted  Japanese brackets
$$
\langle x-x_c \rangle(x)\;:=\;\left(1+|x-x_c|^2\right)^{\frac{1}{2}}\;.
$$
 Since the function $\langle x-x_c \rangle\langle x\rangle^{-1}$ is bounded and invertible one can reformulate the condition \eqref{con_alep1} as follows
\begin{equation*}
 V\;\langle x \rangle^{s}\; \in\; L^\infty(\R)\;,
\end{equation*}
 for some $s>\frac{5}{4}$.
   \hfill $\blacktriangleleft$
\end{remark}

%-----------------------------------%

\section{Scattering theory}\label{sect_scatt}
In this section we will present a class of perturbations $W$ of the operator $T$ for which the wave operators 
\begin{equation}\label{eq-W-op}
			\Omega_\pm(W)\;:=\;{\rm s}-\lim_{t\rightarrow \mp \infty}\expo{\ii (T+W)t}\expo{-\ii Tt}\;.
		\end{equation}
exist and are complete. 
It is worth noticing that in the definition \eqref{eq-W-op} we have tacitly used the fact that the spectrum of $T$ is purely absolutely continuous, and consequently  the spectral projection on the absolutely continuous part of the spectrum of $T$ coincides with the identity, \ie $P_{a.c.}(T)={\bf 1}$.
The proof of the existence and completeness of $\Omega_\pm(W)$ rely on the results of Section \ref{sec_rel_com-pert}
along with the celebrated \emph{Kuroda-Birman theorem} \cite[Theorem XI.9]{reed-simon-III} which, in our specific case, can be stated as follows:
\begin{theorem}[Kuroda-Birman]\label{theo:kb}
Let $W:\R\to\R$ be a  potential such that $T+W$ is self-adjoint. Consider the resolvents 
$$
R_{-\ii}(T)\;:=\;(T+\ii {\bf 1})^{-1}\;,\qquad R_{-\ii}(T+W)\;:=\;(T+W+\ii {\bf 1})^{-1}\;.
$$ If the difference $R_{-\ii}(T)-R_{-\ii}(T+W)$
	is trace-class, then the wave operators $\Omega_\pm(W)$ exist and are complete.
\end{theorem}

\medskip

We are now in position to present our main result concerning 
the scattering theory of the operator $T$.
\begin{proposition}[Scattering theory]\label{theo:waveop}
Let $W:\R\to\R$ be a  potential such that both $W$ and $W':=|W|^{\frac{1}{2}}$	
	 satisfy the conditions of Proposition \ref{theo:hs}. Then, the wave operators $\Omega_\pm(W)$ exist and are complete.
\end{proposition}
\proof
An iterated use of the 
 resolvent identity provides
 $$
\begin{aligned}
	R_{-\ii}(T)-R_{-\ii}(T+W)\;&=\;R_{-\ii}(T)WR_{-\ii}(T+W)\\
	&=\;R_{-\ii}(T)W\big[R_{-\ii}(T)-R_{-\ii}(T)WR_{-\ii}(T+W)\big]\\
	&=\;Z_1+Z_2
\end{aligned}
$$
where the two terms in the last line are given by
$$
Z_1\;:=\;R_{-\ii}(T)WR_{-\ii}(T)\;,\quad Z_2\;:=\;-\big[R_{-\ii}(T)W\big]^2R_{-\ii}(T+W)\;.
$$
In view of  the hypotheses and of Proposition \ref{theo:hs} one has that 
$$
R_{-\ii}(T)W\;=\;\big(WR_{\ii}(T)\big)^*
$$ 
is a Hilbert-Schmidt operator. As a consequence one has that its square, and consequently $Z_2$, are trace-class. The operator 
$Z_1$ can be rewritten as 
$$
Z_1\;=\;\big(W'R_{\ii}(T)\big)^*{\rm sgn}(W)\;W'R_{-\ii}(T)
$$
where ${\rm sgn}(W)$ denotes the bounded multiplicative operator which multiplies elements of $L^(\R)$ by the sign of the function $W$. Again,  from the  hypotheses and of Proposition \ref{theo:hs}
one gets that $W'R_{\pm\ii}(T)$ are Hilbert-Schmidt operators and  
consequently $Z_1$ is trace-class. Since we proved that the difference $R_{-\ii}(T)-R_{-\ii}(T+W)$ is trace-class, the result follows from Theorem \ref{theo:kb}.
\qed

\begin{remark}[Scattering matrix]\label{rk-ScattMat}
The existence of the wave operators $\Omega_\pm(W)$ allows to define the \emph{scattering matrix} 
\begin{equation}
			\bb{S}(W)\;:=\;\Omega_-(W)^*\Omega_+(W)\;.
		\end{equation}
The completeness of the wave operators, along with  $P_{a.c.}(T)={\bf 1}$, ensures that $\bb{S}(W)$ is a unitary operator on $L^2(\R^2)$.
   \hfill $\blacktriangleleft$
\end{remark}

\medskip

The next results provides a class of prototipe bounded potentials which meet the condition stipulated in Proposition \ref{theo:waveop}.
\begin{corollary}\label{cor:scatt1}
The wave operators $\Omega_\pm(W_r)$ associated to the potentials
$$
W_r(x)\;:=\;\left(1+x^2\right)^{-r}\;,\qquad r>\frac{1}{4}
$$ 
exist and are complete.
\end{corollary}
\proof
In view of Corollary \ref{cor:class_HS1}
both $W_r$ and $|W_r|^{\frac{1}{2}}=W_{\frac{r}{2}}$ meet the conditions of Proposition \ref{theo:hs}.
Then, Proposition \ref{theo:waveop} applies.
\qed

\begin{corollary}\label{cor:scatt2}
Let $W$ be a real valued potential such that
$$
W\;\langle x \rangle^s\; \in\; L^\infty(\R)
$$
for some $s>\frac{1}{2}$. 
Then the related  wave operators $\Omega_\pm(W)$ exist and are complete.\end{corollary}
\proof
By hypothesis one has that $W=W_\infty W_{\frac{s}{2}}$
where $W_\infty:= W\langle x \rangle^s\in L^\infty(\R)$ and 
$W_{\frac{s}{2}}=\langle x \rangle^{-s}$.
If $s>\frac{1}{2}$ then $W_{\frac{s}{2}}$ meets the condition of  
Proposition \ref{theo:hs} in view of Corollary \eqref{cor:scatt1}.
Since the multiplication by an essentially bounded function does not change the integrability properties of a given function, it follows that also  $W$ meets the condition of  
Proposition \ref{theo:hs}. Then, Proposition \ref{theo:waveop} applies.
\qed

\begin{remark}[Scattering theory of the thermal Hamiltonian]\label{rk-S-scaT}
The passage from Proposition \ref{theo:waveop} to Theorem \ref{theo:03} is justified by the same argument already discussed in 
 Remark \ref{rk-S-self}  and Remark \ref{rk-S-RC}. 
It is interesting to translate the content of  Corollary \ref{cor:scatt2} in terms of the background potential $V$ which defines the perturbation of the 
 thermal Hamiltonian via the equation \eqref{eq:intro_05}. An analysis similar to that in Remark \ref{rk-S-RC} provides the following result:  If if the  background potential $V:\R\to\R$ meets the condition
\begin{equation*} 
V\;\langle x \rangle^{s}\; \in\; L^\infty(\R)\;,
\end{equation*}
 for some $s>\frac{3}{2}$ then the wave operators defined by  
 \eqref{eq:intro_011} exist and are complete.
   \hfill $\blacktriangleleft$
\end{remark}

%-----------------%
%bibliography
%-----------------%

%------------------%
\end{document}